\documentclass[a4paper,UKenglish]{oasics-v2018}

\usepackage{microtype}
\usepackage{algorithmicx}
\usepackage{algpseudocode}
\usepackage{float}
\floatstyle{boxed} 
\restylefloat{figure}

\title{A Simple Near-Linear Pseudopolynomial Time Randomized Algorithm for Subset~Sum}

\author{Ce Jin}{Institute for Interdisciplinary Information Sciences, Tsinghua University, Beijing, China}{jinc16@mails.tsinghua.edu.cn}{}{}

\author{Hongxun Wu}{Institute for Interdisciplinary Information Sciences, Tsinghua University, Beijing, China}{wuhx18@mails.tsinghua.edu.cn}{}{}

\authorrunning{C.\,Jin and H.\,Wu}

\Copyright{Ce Jin and Hongxun Wu}

\subjclass{
\ccsdesc[100]{Theory of computation~Algorithm design techniques}
}
\keywords{subset sum, formal power series, \textsf{FFT}}

\category{}

\relatedversion{}

\supplement{}

\funding{}

\acknowledgements{The authors would like to thank the anonymous reviewers for their helpful comments.}

\EventEditors{Jeremy Fineman and Michael Mitzenmacher}
\EventNoEds{2}
\EventLongTitle{2nd Symposium on Simplicity in Algorithms (SOSA 2019)}
\EventShortTitle{SOSA 2019}
\EventAcronym{SOSA}
\EventYear{2019}
\EventDate{January 9, 2019}
\EventLocation{Westin San Diego, San Diego, CA, USA}
\EventLogo{}
\SeriesVolume{69}
\ArticleNo{17} 
\nolinenumbers 
\hideOASIcs  

\begin{document}

\maketitle

\begin{abstract}
Given a multiset $S$ of $n$ positive integers and a target integer $t$, the \textsc{Subset~Sum} problem asks to determine whether there exists a subset of $S$ that sums up to $t$. The current best deterministic algorithm, by Koiliaris and Xu [SODA'17], runs in $\tilde O(\sqrt{n}t)$ time, where $\tilde O$ hides poly-logarithm factors. Bringmann [SODA'17] later gave a randomized $\tilde O(n + t)$ time algorithm using two-stage color-coding.  The $\tilde O(n+t)$ running time is believed to be near-optimal.

In this paper, we present a simple and elegant randomized algorithm for \textsc{Subset~Sum} in $\tilde O(n + t)$ time. Our new algorithm actually solves its counting version modulo prime $p>t$, by manipulating generating functions using \textsf{FFT}. 
 \end{abstract}

\newcommand{\fp}{\mathbb{F}_p}
\newcommand{\Q}{\mathbb{Q}}
\newcommand{\R}{\mathbb{R}}
\newcommand{\Z}{\mathbb{Z}}

\section{Introduction}
Given a multiset $S$ of $n$ positive integers and a target integer $t$, the \textsc{Subset~Sum} problem asks to determine whether there exists a subset of $S$ that sums up to $t$.
  It is one of Karp's original NP-complete problems \cite{karp1972reducibility}, and is widely taught in undergraduate algorithm classes. In 1957, Bellman gave the well-known dynamic programming algorithm \cite{bellman1957dynamic} in time $O(nt)$. Pisinger \cite{pisinger1999linear} first improved it to $O(nt / \log{t})$ on word-RAM models.
Recently, Koiliaris and Xu gave a deterministic algorithm \cite{koiliaris2017faster, koiliaris2018subset} in time $\tilde O(\sqrt{n}t)$, which is the best deterministic algorithm so far.
Bringmann \cite{bringmann2017near} later improved the running time to randomized  $\tilde O(n + t)$ using color-coding and layer splitting techniques.
Abboud et al. \cite{abboud2019seth} recently showed that \textsc{Subset~Sum} has no $O(t^{1 - \epsilon}n^{O(1)})$ algorithm for any $\epsilon > 0$, unless the Strong Exponential Time Hypothesis (SETH) is false, so the $\tilde O(n+t)$ time bound is likely to be near-optimal.

In this paper, we present a new randomized algorithm matching the $\tilde O(n + t)$ running time by Bringmann \cite{bringmann2017near}.
The basic idea of our approach is quite straightforward. For prime $p>t$, we give an $\tilde O(n + t)$ algorithm for  $\#_p\textsc{Subset~Sum}$, the counting version of \textsc{Subset~Sum} problem modulo $p$. Then the decision version can be solved with high probability by randomly picking a sufficiently large prime $p$.

A closely related problem is \#\textsc{Knapsack}, which asks for the number of subsets $S$ such that $\sum_{s \in S} s \leq t$. There are extensive studies on approximation algorithms  for the \textsc{\#Knapsack} problem \cite{dyer2003approximate,gopalan2011fptas,rizzi2014faster,gawrychowski2018faster}. 
Our algorithm can solve the modulo $p$ version \textsc{\#$_p$Knapsack} in near-linear pseudopolynomial time for prime $p>t$.  

Compared to the previous near-linear time algorithm for \textsc{Subset~Sum} by Bringmann \cite{bringmann2017near}, our algorithm is simpler and more practical.
The precise running time of our algorithm is $O(n+t\log^2 t)$ with error probability $O((n+t)^{-1})$. 
If a faster algorithm for manipulating formal power series by Brent \cite{brent1976multiple} is applied, it can be improved to $O(n + t\log{t})$ time (see Remark on Lemma \ref{exp}), which is faster than Bringmann's algorithm by a factor of $\log^4 n$. 
\subsection{Main ideas of our algorithm}
The \textsc{Subset~Sum} instance can be encoded as a generating function $A(x)=\prod_{i=1}^n (1+x^{s_i})$, where $s_1,\dots,s_n$ are the input integers, and our goal is to compute the $t$-th coefficient of $A(x)$ and see whether it is zero or not. 

Instead of directly expanding $A(x)$, we consider its logarithm $B(x) = \ln(A(x))$. 
Using basic properties of the logarithm function and its power series, it's possible to compute the first $t+1$  coefficients of $B(x)$ in $\tilde O(t)$ time.
Then we can recover the first $t+1$ coefficients of $A(x)=\exp(B(x))$ in $\tilde O(t)$ time using a simple divide and conquer algorithm with \textsf{FFT} (or a slightly faster algorithm by Brent \cite{brent1976multiple}).

The coefficients involved in the algorithm could be exponentially large. To avoid dealing with high-precision numbers, we pick a prime $p$ and perform arithmetic operations efficiently in the finite field $\fp$, and in the end check whether the result is zero modulo $p$. By picking random $p$ from a large interval, the algorithm succeeds with high probability.

\section{Preliminaries}

\subsection{Subset sum problem}

Given $n$ (not necessarily distinct) positive integers $s_1,s_2,\dots,s_n$  and a target sum $t$, the \textsc{Subset~Sum} problem is to decide whether there exists a subset of indices $I \subseteq \{1,2,\dots,n\}$ such that $\sum_{i \in I}s_i = t$. We also consider the \#$_p$\textsc{Subset~Sum} problem, which asks for the number of such subsets $I$ modulo $p$. We use the word RAM model with word length $w = \Theta (\log t)$ throughout this paper.

\subsection{Polynomials and formal power series}

\paragraph*{Formal power series} 

Let $R[x]$ denote the ring of polynomials over a ring $R$, and $R[[x]]$ denote the ring of formal power series over $R$. A formal power series $f(x) =\sum_{i=0}^\infty f_i x^i$ is a generalization of a polynomial with possibly an infinite number of terms. Polynomial addition and multiplication naturally generalize to $R[[x]]$. Composition  $(f\circ g)(x) = f(g(x)) = \sum_{i=0}^\infty f_i \Big (\sum_{j=1}^\infty g_j x^j\Big )^i$ is  well-defined for $f(x) = \sum_{i=0}^\infty f_i x^i \in R[[x]]$ and $g(x) = \sum_{j=1}^\infty g_j x^j\in xR[[x]]$. Here $xR[[x]]$ (or $xR[x]$) denotes the set of series in $R[[x]]$ (or polynomials in $R[x]$) with zero constant term.
   
\paragraph*{Exponential and logarithm} 

We are familiar with the following two series in $\Q[[x]]$, 
\begin{align}
\ln(1+x) &= \sum_{k=1}^\infty \frac{(-1)^{k-1}x^k}{k},\\
    \exp(x) &= \sum_{k = 0}^\infty \frac{x^k}{k!},
\end{align}
satisfying
\begin{equation}
    \exp\big (\ln(1+f(x))\big) = 1+f(x),
 \end{equation}
 and
 \begin{equation}
  \ln \big( (1+f(x))(1+g(x)) \big ) = \ln(1+f(x)) + \ln(1+g(x))
\end{equation}
for any $f(x), g(x) \in x\Q[x]$. 

\paragraph*{Modulo $x^{t+1}$} 

Our algorithm only deals with the first $t + 1$ terms of any formal power series. 
For $f(x),g(x)\in R[[x]]$, we write $f(x)\equiv g(x) \pmod{x^{t+1}}$ if $[x^i]f(x)=[x^i]g(x)$ for all $0\le i \le t$, where $[x^i]f(x)$ denotes the $i$-th coefficient of $f(x)$.

As an example, define  \begin{equation}
\label{expt}
    \exp_t(x) =\sum_{i=0}^t \frac{x^i}{i!}
\end{equation}
as a $t$-th degree polynomial in $\Q[x]$. Then
$\exp(f(x)) \equiv \exp_t(f(x)) \pmod{x^{t+1}}$
clearly holds for any $f(x) \in x\Q[[x]]$.

\subsection{Modulo prime \texorpdfstring{$p$}{p}}
\label{modprime}
\newcommand{\overbar}[1]{\overline{\mkern-1.0mu#1\mkern-1.5mu}}
To avoid dealing with large fractions or floating-point numbers, we will work in the finite field $\fp = \{\overbar 0, \overbar 1,\dots, \overbar{p-1}\}$ of prime order $p=2^{\Theta(\log t)}$.  
Addition and multiplication in $\fp$ take $O(1)$ time in the word RAM model. Finding the multiplicative inverse of a nonzero element in $\fp$ takes $O(\log p)$ time using extended Euclidean algorithm \cite[Section~31.2]{cormen2009introduction}.

Our algorithm will regard polynomial coefficients as elements from $\fp$. The coefficients can be rational numbers, but their denominators should not have prime factor $p$.
Formally, let 
\begin{equation}\Z_{p\Z} = \{r/s \in \Q: \text{$r,s$ are coprime integers, $p$ does not divide $s$}\}
\end{equation}
and apply the canonical homomorphism from $\Z_{p\Z}[x]$ to $\fp[x]$, determined by
\begin{align}
    r/s \mapsto \bar{s}^{-1}\bar{r},\;\; x \mapsto x.
\end{align}
We use $\bar{A}$ or $A \bmod p$ to denote $A$'s image in $\fp[x]$.

From now on we assume $p>t$, so that $\exp_t(x) \in \Z_{p\Z}[x]$ (see equation (\ref{expt})), and  let $\overline{\exp_t}(x)$ denote its image in $\fp[x]$. 

\subsection{Computing exponential using FFT}
\begin{lemma}[FFT]
\label{fft}
Given two polynomials $f(x),g(x) \in \fp[x]$ of degree at most $t$, one can compute their product $f(x)g(x)$ in $O(t \log t)$ time.
\end{lemma}

\begin{proof}
The classic \textsf{FFT} algorithm \cite[Chapter~30] {cormen2009introduction} can multiply $f(x)$ and $g(x)$, regarded as polynomials in $\Z[x]$, in $O(t\log t)$ time. Then take the remainder of each coefficient modulo $p$.
\end{proof}
Lemma \ref{exp} is a classical result on manipulating formal power series, and is the main building block of our algorithm.

\begin{lemma}[Brent  \cite{brent1976multiple}]
\label{exp}
Given a polynomial $f(x) \in x\fp[x]$ of degree at most $t$ $(t<p)$, one can compute  a polynomial $g(x) \in \fp[x]$ in $\tilde O(t)$ time such that $g(x) \equiv \overline{\exp_t}(f(x)) \pmod{x^{t+1}}$.

\end{lemma}
\begin{remark}
Brent's algorithm  \cite{brent1976multiple} uses Newton's iterative method and runs in time $O(t\log t)$. Here we describe a simpler $O(t\log^2 t)$ algorithm by standard divide and conquer. We present the algorithm as over $\Q$ for notational simplicity.
\end{remark}
\begin{proof}
Let $f(x) = \sum_{i=1}^t f_i x^i$ and $g(x) = \exp(f(x)) = \sum_{i=0}^\infty g_i x^i$. Then $g'(x) = g(x) f'(x)$. Comparing the $(i-1)$-th coefficients on both sides gives a recurrence relation 
\begin{equation}
\label{recurrence}
    g_i = i^{-1}\sum_{j=0}^{i-1}(i-j)f_{i-j}g_j 
\end{equation}
with initial value $g_0 = 1$.  The desired coefficients $g_1,\dots,g_t$ can be computed using the algorithm in Figure \ref{computing}, which simply reorganizes the computation of recurrence formula (\ref{recurrence}) as a recursion.

To speed up this algorithm, define polynomial $F(x) = \sum_{k=0}^{r-l} kf_k x^k, G(x) = \sum_{j=0}^{m-l} g_{j+l}x^j$ and use \textsf{FFT} to compute $H(x) = F(x)G(x)$ in $O((r-l)\log (r-l))$ time after \textsc{Compute}$(l,m)$ returns. Then $\sum_{j=l}^m (i-j)f_{i-j} g_j  = [x^{i-l}]H(x)$, and hence the \textbf{for} loop runs in $O(r-m)$ time. The total running time is $T(t) = 2T(t/2) + O(t\log t) = O(t\log ^2 t)$.

\begin{figure}
\centering
\begin{algorithmic}
\Procedure{Compute}{$l,r$}
\Comment{after \textsc{Compute}$(l,r)$ returns, all values $g_1,\dots,g_r$ are ready}
\If{$l<r$}
    \State{$m \gets \lfloor (l+r)/2 \rfloor$}
    \State{\textsc{Compute}($l,m$)}
    \For{$i \gets m+1,m+2,\dots,r$} 
            \State{$g_i \gets g_i + i^{-1}\sum_{j=l}^m  (i-j)f_{i-j}g_j$}
    \EndFor
    \State{\textsc{Compute}($m+1,r$)}
\EndIf
\EndProcedure
\\
\Procedure{Main}{}
\State{Initialize $g_0 \gets 1, g_i \gets 0 (1\le i \le t)$}
\State{\textsc{Compute}$(0,t)$}
\EndProcedure
\end{algorithmic}
\caption{Algorithm for computing $g_1,\dots,g_t$}
\label{computing}
\end{figure}
\end{proof}

\section{Main algorithm} \label{algo}

Recall that we are given $n$ positive integers $s_1,\dots,s_n$ and a target sum $t$.
Consider the generating function $A(x)$ defined by
\begin{equation} A(x) = \prod_{i=1}^n (1+x^{s_i}). \end{equation}
The number of subsets that sum up to $t$ is $[x^t]A(x)$. The \textsc{Subset~Sum} instance has a solution if and only if $[x^t]A(x)\neq 0$.

\begin{lemma} \label{lemprob} 
Suppose $[x^t]A(x) \neq 0$. Let $p$ be a uniform random prime from $[t + 1, (n + t)^{3}]$. With probability $1-O((n+t)^{-1})$, $p$ does not divide $[x^t]A(x)$.  
\end{lemma}
\begin{proof}
Notice that $[x^t] A(x) \leq 2^n$, so it has at most $n$ prime factors. Since there are $\Omega((n+t)^2)$ primes in the interval, the probability that $p$ divides $[x^t]A(x)$ is $O((n+t)^{-1})$.
\end{proof}

\begin{lemma}
\label{lemln}
Let $B(x) = \ln(A(x)) \in \Q[[x]]$. For prime $p > t$, in $\tilde O(t)$ time one can compute $([x^r]B(x)) \bmod p$ for all $0\le r \le t$.
\end{lemma}
\begin{proof}
By definition of $B(x)$, 
\begin{align}
    B(x) = \ln \Big ( \prod_{i=1}^n (1+ x^{s_i})\Big )  = \sum_{i=1}^n \ln (1+x^{s_i})  = \sum_{i=1}^n \sum_{j=1}^\infty \frac{(-1)^{j-1}}{j}x^{s_ij}. 
\end{align}
Let $a_k$ be the size of the set $\{j:s_j=k\}$, and define polynomial
\begin{align}
\label{bx}
     B_t(x)= \sum_{i=1}^n \sum_{j=1}^{\lfloor t/s_i \rfloor}\frac{(-1)^{j-1}}{j}x^{s_ij}
     = \sum_{k=1}^t \sum_{j=1}^{\lfloor t/k \rfloor}\frac{a_k (-1)^{j-1}}{j}x^{jk}.
\end{align}
Then $[x^r]B_t(x) = [x^r]B(x)$ for all $0\le r \le t$.

Note that the denominators $j$ in (\ref{bx}) do not have prime factor $p$.
 After preparing the multiplicative inverses $\bar{j}^{-1}$ for each $1\le j \le t$, we can compute all $([x^r]B_t(x)) \bmod p$ by simply iterating over $k,j$ in equation (\ref{bx}), which only takes $\sum_{k=1}^t \lfloor t/k \rfloor  = O(t\log t)$ time.
\end{proof}
\begin{lemma}
\label{lemalgo}
For prime $p > t$, one can compute $([x^r]A(x)) \bmod{p}$ for all $0\le r \le t$ in $\tilde O(t)$ time.
\end{lemma}
\begin{proof}
 Let $B(x) = \ln (A(x))$. Then $A(x) = \exp(B(x)) \equiv \exp_t(B_t(x)) \pmod{x^{t+1}}$, where $B_t(x) = \sum_{i=0}^t ([x^i]B(x))x^i$.
 We use Lemma \ref{lemln} to compute $B_t(x)$'s image $\overbar{B_t}(x) \in \fp[x]$, and then use Lemma \ref{exp} to compute the first $t+1$ terms of $\overbar{\exp_t} ( \overbar{B_t}(x)) $, which give the values of  $([x^r]A(x)) \bmod p$ for all $0\le r \le t$.
\end{proof}

\begin{theorem}
The \textsc{Subset~Sum} problem can be solved in time $\tilde O(n+t)$ by a randomized algorithm with one-sided error probability $O((n+t)^{-1})$.
\end{theorem}
\begin{proof}
By sampling and using Miller-Rabin primality test \cite[Section~31.8]{cormen2009introduction}, we can pick a uniform random prime $p$ from interval $[t+1,(n+t)^3]$   in $(\log (n+t))^{O(1)}$ time with $O((n+t)^{-1})$ failure probability.  Then the theorem immediately follows from Lemma \ref{lemprob} and Lemma \ref{lemalgo}.
\end{proof}

\bibliography{oasics-v2018-sample-article}

\end{document}